\documentclass{article}%
\usepackage{amsfonts}
\usepackage{amsmath}
\usepackage{amssymb}
\usepackage{graphicx}%
\providecommand{\U}[1]{\protect\rule{.1in}{.1in}}
\newtheorem{theorem}{Theorem} [section]

\newtheorem{conjecture}[theorem]{Conjecture}
\newtheorem{corollary}[theorem]{Corollary}

\newtheorem{lemma}[theorem]{Lemma}

\newtheorem{problem}[theorem]{Problem}
\newtheorem{proposition}[theorem]{Proposition}

\newenvironment{proof}[1][Proof]{\noindent\textbf{#1.} }{\ \rule{0.5em}{0.5em}}
\begin{document}

\title{On $f$-Symmetries of the Independence Polynomial}
\author{Vadim E. Levit\\Department of Computer Science and Mathematics\\Ariel University, Israel\\levitv@ariel.ac.il
\and Eugen Mandrescu\\Department of Computer Science\\Holon Institute of Technology, Israel\\eugen\_m@hit.ac.il}
\date{}
\maketitle

\begin{abstract}
An \textit{independent} set in a graph is a set of pairwise non-adjacent
vertices, and $\alpha(G)$ is the size of a maximum independent set in the
graph $G$.

If $s_{k}$ is the number of independent sets of cardinality $k$ in $G$, then
\[
I(G;x)=s_{0}+s_{1}x+s_{2}x^{2}+...+s_{\alpha}x^{\alpha},\alpha=\alpha\left(
G\right)  ,
\]
is called the \textit{independence polynomial} of $G$ (I. Gutman and F.
Harary, 1983).

If $s_{\alpha-i}=f\left(  i\right)  \cdot s_{i}$ holds for every $i\in\left\{
0,1,...,\left\lfloor \alpha/2\right\rfloor \right\}  $, then $I(G;x)$ is
called $f$-\textit{symmetric (}$f$-\textit{palindromic)}. If $f\left(
i\right)  =1,i\in\left\{  0,1,...,\left\lfloor \alpha/2\right\rfloor \right\}
$, then $I(G;x)$ is called \textit{symmetric (palindromic).}

The \textit{corona} of the graphs $G$ and $H$ is the graph $G\circ H$ obtained
by joining each vertex of $G$ to all the vertices of a copy of $H$.

In this paper we show that if $H$ is a graph with $p$ vertices, $q$ edges, and
$\alpha\left(  H\right)  =2$, then $I\left(  G\circ H;x\right)  $ is
$f$-symmetric, where
\[
f\left(  i\right)  =\left(  \frac{p\left(  p-1\right)  }{2}-q\right)
^{\frac{\alpha}{2}-i},0\leq i\leq\alpha=\alpha\left(  G\circ H\right)  .
\]
In particular, if $H=K_{r}-e,r\geq2$, we show that $I\left(  G\circ
H;x\right)  $ is symmetric and unimodal, with a unique mode. This finding
generalizes results due to Stevanovi\'{c} \cite{St98} and Mandrescu
\cite{Mandrescu2012} claiming that $I\left(  G\circ\left(  K_{2}-e\right)
;x\right)  =I\left(  G\circ2K_{1};x\right)  $ is symmetric and unimodal for
every graph $G$.\smallskip

\smallskip


\end{abstract}

\section{Introduction}

Throughout this paper $G$ is a simple graph with vertex set $V(G)$ and edge
set $E(G)$. If $X\subset V(G)$, then $G[X]$ is the subgraph of $G$ spanned by
$X$.

By $G-W$ we mean the subgraph $G[V-W]$, if $W\subset V(G)$. We also denote by
$G-F$ the subgraph of $G$ obtained by deleting the edges of $F$, for $F\subset
E(G)$, and we write shortly $G-e$, whenever $F$ $=\{e\}$. 

The \textit{neighborhood} of a vertex $v\in V$ is the set
\begin{align*}
N_{G}(v)  & =\{w:w\in V\text{\textit{ and }}vw\in E\},\text{ and}\\
N_{G}[v]  & =N_{G}(v)\cup\{v\}.
\end{align*}
If there is no ambiguity on $G$, we use $N(v)$ and $N[v]$, respectively. 

$K_{n},P_{n},C_{n}$ denote respectively, the complete graph on $n\geq1$
vertices, the chordless path on $n\geq1$ vertices, and the chordless cycle on
$n\geq3$ vertices.

The \textit{disjoint union} of the graphs $G_{1}$, $G_{2}$ is the graph
$G=G_{1}\cup G_{2}$ having as vertex set the disjoint union of $V(G_{1})$,
$V(G_{2})$, and as edge set the disjoint union of $E(G_{1})$, $E(G_{2})$. In
particular, $nG$ denotes the disjoint union of $n>1$ copies of the graph $G$.
The \textit{Zykov sum} of the disjoint graphs $G_{1}$, $G_{2}$ is the graph
$G_{1}+G_{2}$ with $V(G_{1})\cup V(G_{2})$ as a vertex set and
\[
E(G_{1})\cup E(G_{2})\cup\{v_{1}v_{2}:v_{1}\in V(G_{1}),v_{2}\in V(G_{2})\}
\]
as an edge set. 

The \textit{corona} of the graphs $G$ and $H$ is the graph $G\circ H$ obtained
from $G$ and $\left\vert V\left(  G\right)  \right\vert $ copies of $H$, such
that each vertex of $G$ is joined to all vertices of a copy of $H$.

An \textit{independent} set in $G$ is a set of pairwise non-adjacent vertices.
An independent set of maximum size is a \textit{maximum independent set} of
$G$, and the \textit{independence number }$\alpha(G)$ is the cardinality of a
maximum independent set in $G$.

Let $s_{k}$ be the number of independent sets of size $k$ in a graph $G$. The
polynomial
\[
I(G;x)=s_{0}+s_{1}x+s_{2}x^{2}+...+s_{\alpha}x^{\alpha},\quad\alpha
=\alpha\left(  G\right)  ,
\]
is called the \textit{independence polynomial} of $G$ \cite{GuHa83}. For a
survey on independence polynomials of graphs see \cite{LevManGreece}. Some
basic procedures to compute the independence polynomial of a graph are
recalled in the following.

\begin{theorem}
\cite{GuHa83} \label{th1}\emph{(i)} $I(G_{1}\cup G_{2};x)=I(G_{1};x)\cdot
I(G_{2};x)$;

\emph{(ii)} $I(G_{1}+G_{2};x)=I(G_{1};x)+I(G_{2};x)-1$;

\emph{(iii)} $I(G;x)=I(G-v;x)+x\cdot I(G-N[v];x)$ holds for every $v\in V(G)$.
\end{theorem}

A finite sequence of real numbers $(a_{0},a_{1},a_{2},...,a_{n})$ is said to be:

\begin{itemize}
\item \textit{unimodal} if there exists an index $k\in\{0,1,...,n\}$, called
the \textit{mode} of the sequence, such that
\[
a_{0}\leq...\leq a_{k-1}\leq a_{k}\geq a_{k+1}\geq...\geq a_{n};
\]

\item $f$\textit{-symmetric} ($f$\textit{-palindromic}) if $a_{n-i}=f\left(
i\right)  \cdot a_{i}$ for all $i\in\left\{  0,...,\left\lfloor
n/2\right\rfloor \right\}  $;

\item \textit{symmetric} (\textit{palindromic}) if $a_{i}=a_{n-i}%
,i=0,1,...,\left\lfloor n/2\right\rfloor $, i.e., $f\left(  i\right)  =1$ for
all $i\in\left\{  0,...,\left\lfloor n/2\right\rfloor \right\}  $.
\end{itemize}

A polynomial is called \textit{unimodal (symmetric, }$f$\textit{-symmetric)}
if the sequence of its coefficients is unimodal (symmetric, $f$-symmetric,
respectively). For instance, the independence polynomial:

\begin{itemize}
\item $I(K_{127}+3K_{7};x)=1+\mathbf{148}x+147x^{2}+\mathbf{343}x^{3}$ is non-unimodal;

\item $I(K_{43}+3K_{7};x)=1+64x+147x^{2}+\mathbf{343}x^{3}$ is unimodal and non-symmetric;

\item $I(K_{18}+3K_{3}+4K_{1};x)=1+\allowbreak31x+\mathbf{33}x^{2}%
+31x^{3}+x^{4}$ is symmetric and unimodal;

\item $I(K_{52}+3K_{4}+4K_{1};x)=1+68x+\mathbf{54}x^{2}+68x^{3}+x^{4}$ is
symmetric and non-unimodal;

\item $I\left(  P_{3}\circ\left(  K_{2}\cup K_{1}\right)  ;x\right)
=1+12x+52x^{2}+105x^{3}+104x^{4}+48x^{5}+8x^{6}$ is $f$-symmetric for
$f(i)=2^{3-i},0\leq i\leq3$.
\end{itemize}

For other examples, see \cite{LeMa03b,LeMa03c,LeMa04c,LeMa04b,Wang,Zfu}.
Alavi, Malde, Schwenk and Erd\"{o}s proved that for every permutation $\pi$ of
$\{1,2,...,\alpha\}$ there is a graph $G$ with $\alpha(G)=\alpha$ such that
$s_{\pi(1)}<s_{\pi(2)}<...<s_{\pi(\alpha)}$ \cite{AlMaScEr87}.

\begin{theorem}
\label{th2}\cite{Gu92d} $I\left(  G\circ H;x\right)  =\left(  I\left(
H;x\right)  \right)  ^{n}\bullet I\left(  G;\frac{x}{I\left(  H;x\right)
}\right)  $, where $n=\left\vert V\left(  G\right)  \right\vert $.
\end{theorem}

The symmetry of matching polynomial and characteristic polynomial of a graph
were examined in \cite{Kennedy}, while for independence polynomial we quote
\cite{Gu93,St98,LevMan2007,LevMan2008}. It is known that the product of two
unimodal polynomials is not necessarily unimodal.

\begin{theorem}
\label{th4} \cite{Andrews} If $P$ and $Q$ are both unimodal and symmetric,
then $P\cdot Q$ is unimodal and symmetric.
\end{theorem}

However, the above result can not be generalized to the case when $P$ is
unimodal and symmetric, while $Q$ is unimodal and non-symmetric; e.g.,
\begin{align*}
P  &  =1+x+3x^{2}+x^{3}+x^{4},Q=1+x+x^{2}+x^{3}+2x^{4}\text{, while}\\
P\cdot Q  &  =1+2x+5x^{2}+6x^{3}+8x^{4}+7x^{5}+8x^{6}+3x^{7}+2x^{8}.
\end{align*}

It is worth mentioning that one can produce graphs with symmetric independence
polynomials in different ways (see, for instance, \cite{BS2010,Gu92c,St98}).

In this paper we prove that if $H$ is a graph with $p$ vertices, $q$ edges,
and $\alpha\left(  H\right)  =2$, then $I(G\circ H;x)$ is $f$-symmetric, where
$f\left(  i\right)  =\left(  \frac{p\left(  p-1\right)  }{2}-q\right)
^{\frac{\alpha}{2}-i},0\leq i\leq\alpha=\alpha\left(  G\circ H\right)  $. In
particular, if $H=K_{r}-e$, where $r\geq2$ and $e$ is an edge of $K_{r}$, then
$I\left(  G\circ H;x\right)  $ is symmetric and unimodal with a unique mode.
As a consequence, we deduce that $I\left(  G\circ\left(  K_{2}-e\right)
;x\right)  =I\left(  G\circ2K_{1};x\right)  $ is both symmetric \cite{St98}
and unimodal \cite{Mandrescu2012} for every graph $G$.

\section{Results}

The polynomial $P(x)$ is symmetric if and only if it equals its reciprocal,
i.e.,
\[
P(x)=x^{\deg(P)}\cdot P\left(  \frac{1}{x}\right)  .
\]
We generalize this observation using $\frac{1}{cx}$ instead of $\frac{1}{x}$.

\begin{lemma}
\label{Lemma}If $P\left(  x\right)  =\sum\limits_{i=0}^{n}a_{i}x^{i}$ is a
polynomial of degree $n$, then
\[
P(x)=c^{\frac{n}{2}}\cdot x^{n}\cdot P\left(  \frac{1}{c\cdot x}\right)
\ \text{if\ and\ only\ if}\ a_{n-i}=c^{\frac{n}{2}-i}\cdot a_{i},0\leq i\leq
n.
\]

\end{lemma}

\begin{proof}
Since
\[
c^{\frac{n}{2}}\cdot x^{n}\cdot P\left(  \frac{1}{cx}\right)  =c^{\frac{n}{2}%
}\cdot x^{n}\cdot\sum\limits_{i=0}^{n}\frac{a_{i}}{\left(  cx\right)  ^{i}%
}=\sum\limits_{i=0}^{n}c^{\frac{n}{2}-i}\cdot a_{i}\cdot x^{n-i}%
=\sum\limits_{i=0}^{n}c^{i-\frac{n}{2}}\cdot a_{n-i}\cdot x^{i},
\]
we infer that
\[
P(x)=c^{\frac{n}{2}}\cdot x^{n}\cdot P\left(  \frac{1}{cx}\right)
\Leftrightarrow a_{i}=c^{i-\frac{n}{2}}\cdot a_{n-i}\Leftrightarrow
a_{n-i}=c^{\frac{n}{2}-i}\cdot a_{i},0\leq i\leq n,
\]
and this completes the proof.
\end{proof}

\bigskip If $\frac{a\left(  x\right)  }{b\left(  x\right)  }=\frac{a\left(
f\left(  x\right)  \right)  }{b\left(  f\left(  x\right)  \right)  }$, then
$f\left(  x\right)  $ is an \textit{invariant} of the rational function
$\frac{a\left(  x\right)  }{b\left(  x\right)  }$. Actually, from the point of
view of Gutman's formula from Theorem \ref{th2}, we are mostly interested in
finding invariants for rational functions of the form $\frac{x}{b\left(
x\right)  }$.

\begin{lemma}
\label{lem3}The rational function $\frac{x}{b_{0}+b_{1}x+b_{2}x^{2}}$ admits
only two invariants, namely, $f_{1}\left(  x\right)  =x$ and $f_{2}\left(
x\right)  =\frac{b_{0}}{b_{2}x}$. \ 
\end{lemma}

\begin{proof}
Since%
\[
\frac{x}{b_{0}+b_{1}x+b_{2}x^{2}}=\frac{f\left(  x\right)  }{b_{0}%
+b_{1}f\left(  x\right)  +b_{2}f\left(  x\right)  ^{2}}\Leftrightarrow
b_{2}xf\left(  x\right)  ^{2}-\left(  b_{0}+b_{2}x^{2}\right)  f\left(
x\right)  +b_{0}x=0,
\]
we get the following solutions: $f_{1}\left(  x\right)  =x$ and $f_{2}\left(
x\right)  =\frac{b_{0}}{b_{2}x}$.
\end{proof}

\begin{theorem}
\label{th3}Let $G$ be a graph of order $n$, and $H$ be a graph of order $p$
and size $q$, with $\alpha\left(  H\right)  =2$. Then the polynomial $I(G\circ
H;x)$ is:

\emph{(i)} $f$-symmetric, where $f\left(  i\right)  =\left(  \frac{p\left(
p-1\right)  }{2}-q\right)  ^{\frac{\alpha}{2}-i},0\leq i\leq\alpha
=\alpha\left(  G\circ H\right)  $;

\emph{(ii)} symmetric if and only if $H=K_{r}-e$ for some $r\geq2$.
\end{theorem}

\begin{proof}
\emph{(i)} Since $\alpha\left(  H\right)  =2$, one can write $I\left(
H;x\right)  =1+px+mx^{2}$, where $m=\frac{p\left(  p-1\right)  }{2}-q$. By
Lemma \ref{lem3}, the function $g(x)=\left(  mx\right)  ^{-1}$ is the only
non-trivial invariant of $\frac{x}{I\left(  H;x\right)  }$. Thus we get
\[
I\left(  H;g\left(  x\right)  \right)  =1+p\cdot g\left(  x\right)
+m\cdot\left(  g\left(  x\right)  \right)  ^{2}=\left(  mx^{2}\right)
^{-1}\cdot I(H;x).
\]
According to Lemma \ref{lem3}, it follows
\begin{gather*}
I\left(  G\circ H;g(x)\right)  =\left(  I(H;g\left(  x\right)  )\right)
^{n}\cdot I\left(  G;\frac{g\left(  x\right)  }{I(H;g\left(  x\right)
)}\right)  =\\
=\left(  mx^{2}\right)  ^{-n}\cdot\left(  I(H;x)\right)  ^{n}\cdot I\left(
G;\frac{x}{I(H;x)}\right)  =\left(  mx^{2}\right)  ^{-n}\cdot I\left(  G\circ
H;x\right)  .
\end{gather*}
Consequently, we have $I\left(  G\circ H;x\right)  =m^{n}\cdot x^{2n}\cdot
I\left(  G\circ H;\frac{1}{mx}\right)  $. Since every $v\in V\left(  G\right)
$ is joined in $G\circ H$ to a copy of $H$, it follows that each independent
set $S$ of $G\circ H$ containing a pair of non-adjacent vertices from every
copy of $H$ is a maximum independent in $G\circ H$, i.e., $\alpha\left(
G\circ H\right)  =2n$. Lemma \ref{Lemma} implies that $I(G\circ H;x)$ is
$f$-symmetric, where
\[
f\left(  i\right)  =\left(  \frac{p\left(  p-1\right)  }{2}-q\right)
^{n-i},0\leq i\leq2n=\deg(I\left(  G\circ H;x\right)  ).
\]

\emph{(ii)} The polynomial $I\left(  G\circ H;x\right)  $ is symmetric if and
only if $f\left(  i\right)  =1$, $0\leq i\leq\alpha$. By part \emph{(i)}, it
means $\frac{p\left(  p-1\right)  }{2}-q=1$, and this holds if and only if
$H=K_{p}-e$, where $e$ is some edge of $K_{p}$.
\end{proof}

It is worth noticing that $K_{2}-e=2K_{1}$, and this leads to the following.

\begin{corollary}
\cite{St98} The polynomial $I\left(  G\circ2K_{1};x\right)  $ is symmetric for
every graph G.
\end{corollary}

Recall that a graph $G$ is \textit{perfect} if $\chi(H)=\omega(H)$ for every
induced subgraph $H$ of $G$, where $\chi(H)$ denotes the chromatic number of
$H$ \cite{Berge0}.

\begin{proposition}
\cite{LevMan2007b}\label{prop1} If $G$ is a perfect graph with $\alpha
(G)=\alpha$ and $\omega=\omega(G)$, then
\[
s_{\left\lceil (\omega\alpha-1)/\left(  \omega+1\right)  \right\rceil }%
\geq.....\geq s_{\alpha-1}\geq s_{\alpha}.
\]

\end{proposition}

The Strong Perfect Graph Theorem, due to\emph{\ }Chudnovsky \textit{et al}.,
\cite{ChRoSeTh03}, \cite{ChRoSeTh04}, asserts that a graph $G$ is perfect if
and only if it contains no odd hole (i.e., $C_{2n+1},n\geq2$) and no odd
antihole (i.e., $\overline{C_{2n+1}},n\geq2$) as an induced subgraph.

\begin{proposition}
If $G$ is a perfect graph of order $n$, then the coefficients $\left(
s_{i}\right)  $ of the polynomial $I\left(  G\circ\left(  K_{p}\cup
K_{q}\right)  ;x\right)  $ satisfy the following:
\[
s_{\left\lceil (2n\omega-1)/\left(  \omega+1\right)  \right\rceil }%
\geq.....\geq s_{2n-1}\geq s_{2n}\ \text{and }s_{0}\leq.....\leq s_{t-1}\leq
s_{t},
\]%
\[
\ \text{where }t=2n-\left\lceil \left(  2n\omega-1\right)  /\left(
\omega+1\right)  \right\rceil \,\omega=\max\left\{  \omega\left(  G\right)
,p,q\right\}  .
\]

\end{proposition}

\begin{proof}
It is easy to see that
\[
\alpha\left(  G\circ\left(  K_{p}\cup K_{q}\right)  \right)  =2n,\quad
\omega=\omega\left(  G\circ\left(  K_{p}\cup K_{q}\right)  \right)
=\max\left\{  \omega\left(  G\right)  ,p,q\right\}  .
\]
Since $G\circ\left(  K_{p}\cup K_{q}\right)  $ has no odd hole and no odd
antihole as an induced subgraph, Strong Perfect Graph Theorem assures that
$G\circ\left(  K_{p}\cup K_{q}\right)  $ is perfect. According to Proposition
\ref{prop1}, it follows that
\[
s_{\left\lceil (2n\omega-1)/\left(  \omega+1\right)  \right\rceil }%
\geq.....\geq s_{2n-1}\geq s_{2n}.
\]
By Theorem \ref{th3}, we have that
\[
s_{2n-i}=\left(  pq\right)  ^{n-i}\cdot s_{i},0\leq i\leq n.
\]
Since $s_{2n-i}\geq s_{2n-i+1}$, for $2n-i\geq$ $\left\lceil \left(
2n\omega-1\right)  /\left(  \omega+1\right)  \right\rceil $, we obtain
\[
\left(  pq\right)  ^{n-i}\cdot s_{i}\geq\left(  pq\right)  ^{n-i+1}\cdot
s_{i-1}\Longleftrightarrow s_{i}\geq pq\cdot s_{i-1}\Longrightarrow
s_{i-1}\leq s_{i},
\]
for $1\leq i\leq2n-\left\lceil \left(  2n\omega-1\right)  /\left(
\omega+1\right)  \right\rceil $, as claimed.
\end{proof}

It is easy to see that the sum of two symmetric and unimodal polynomials is
not necessarily symmetric and/or unimodal.

\begin{lemma}
\label{lem1}Let $p(x)$ and $q(x)$ be polynomials of degree $r$ and $r-1$
respectively, for some $r\geq2$, and let $p(0)\neq0$ and $q(0)=0$. If $p(x)$
and $q(x)$ are symmetric and unimodal, then so is $p(x)+q(x)$. Moreover, if
the mode of $p\left(  x\right)  $ or $q(x)$ is unique, then the mode of
$p(x)+q(x)$ is unique as well.
\end{lemma}

\begin{proof}
The symmetry and unimodality of $p(x)+q(x)$ were proved in \cite{BS2010}.

Assume that the mode of $p(x)=a_{0}+...+a_{s-1}x^{s-1}+a_{s}x^{s}%
+a_{s-1}x^{s-1}+...+a_{0}x^{r}$ is equal to $s$ and is unique, i.e.,
$a_{0}\leq a_{1}\leq...\leq a_{s-1}<a_{s}>a_{s-1}\geq...\geq a_{0}$, where
$r=2s$. The polynomial $q(x)$ is symmetric, unimodal, $\deg q=r-1$, and
$q(0)=0$, that is $q\left(  x\right)  =b_{1}x+...+b_{1}x^{r-1}$ satisfies
$b_{1}\leq...\leq b_{s-1}\leq b_{s}\geq b_{s-1}\geq...\geq b_{1}$. Then
$p(x)+q(x)=a_{0}+\left(  a_{1}+b_{1}\right)  x+...+\left(  a_{r-1}%
+b_{r-1}\right)  x^{r-1}+a_{0}x^{r}$ is symmetric, unimodal, and
$a_{s-1}+b_{s-1}<a_{s}+b_{s}$, i.e., its mode is equal to $s$ and it is unique.

Similarly, one can show that the mode of $p(x)+q(x)$ is unique, whenever the
mode of $q(x)$ is unique.
\end{proof}

\begin{lemma}
\label{lem2}If $a>1$, and $P=1+s_{1}x+...+s_{n-1}x^{n-1}+s_{n}x^{n}%
+s_{n-1}x^{n+1}+...+x^{2n}$ is symmetric, unimodal with a unique mode, then
$Q=\left(  1+ax+x^{2}\right)  \cdot P$ is symmetric and unimodal with a unique
mode, equal to $n+1$.
\end{lemma}

\begin{proof}
The symmetry of $Q$ follows from Theorem \ref{th4}. The coefficients of
$x^{n}$, $x^{n+1}$, $x^{n+2}$ in $Q$ are respectively, $t_{n}=t_{n+2}%
=s_{n}+as_{n-1}+s_{n-2}$ and $t_{n+1}=2s_{n-1}+as_{n}$. Hence we obtain that
\[
t_{n+1}-t_{n}=t_{n+1}-t_{n+2}=\left(  s_{n-1}-s_{n-2}\right)  +\left(
a-1\right)  \left(  s_{n}-s_{n-1}\right)  >0,
\]
which implies that the mode of $Q$ is equal to $n+1$ and it is unique.
\end{proof}

\begin{theorem}
\label{Th}If $H=K_{r}-e,r\geq2$, then the polynomial $I(G\circ H;x)$ is
unimodal and symmetric for every graph $G$. Moreover, the mode of $I(G\circ
H;x)$ is unique and equal to the order of $G$.
\end{theorem}

\begin{proof}
The polynomial $I(G\circ H;x)$ is symmetric, according to Theorem
\ref{th3}\emph{(ii)}.

We show, by induction on the order $n=\left\vert V\left(  G\right)
\right\vert $ of $G$, that $I(G\circ H;x)$ is unimodal and its mode is unique
and equal to $n$.

If $n=1$, then $G=K_{1}$ and $I(G\circ H;x)=I(K_{r+1}-e;x)=1+\left(
r+1\right)  x+x^{2}$, which is clearly unimodal and the mode is unique and
equal to $n$.

If $n=2$, then either $G=2K_{1}$ and
\begin{gather*}
I(G\circ H;x)=I(2K_{1}\circ H;x)=\left(  1+\left(  r+1\right)  x+x^{2}\right)
^{2}=\\
=\allowbreak\allowbreak1+\left(  2r+2\right)  x+\left(  r^{2}+2r+3\right)
x^{2}+\left(  2r+2\right)  x^{3}+x^{4},
\end{gather*}
or $G=K_{2}=(\{v_{1},v_{2}\},\{v_{1}v_{2}\})$ and
\begin{gather*}
I(G\circ H;x)=I(K_{2}\circ H-v_{1};x)+x\cdot I(K_{2}\circ H-N[v_{1}];x)=\\
=I(H;x)\cdot I(K_{r+1};x)+x\cdot I(H;x)=\left(  1+\left(  r+1\right)
x+x^{2}\right)  \cdot\left(  1+\left(  r+2\right)  x+x^{2}\right)  =\\
=1+\allowbreak\left(  2r+3\right)  x+\left(  r^{2}+3r+4\right)  x^{2}+\left(
2r+3\right)  x^{3}+x^{4}.
\end{gather*}
In both cases, $I(G\circ H;x)$ is clearly unimodal and the mode is unique and
equals $n$.

Let $G$ be a graph of order $n\geq3$.

Clearly, if $E=\emptyset$, then $I(G\circ H;x)=\left(  1+\left(  r+1\right)
x+x^{2}\right)  ^{n}$, which is unimodal, according to Theorem \ref{th4}, and
its mode is unique and equal to $n$, by Lemma \ref{lem2}.

Suppose that $E\neq\emptyset$, and let $v\in V$ be with $N_{G}\left(
v\right)  =\left\{  u_{i}:1\leq i\leq k\right\}  $.

Applying Theorem \ref{th3}, we obtain
\[
I(G\circ H;x)=I(G\circ H-v;x)+x\bullet I(G\circ H-N_{G\circ H}\left[
v\right]  ;x)=p\left(  x\right)  +q\left(  x\right)  ,
\]
where
\[
p\left(  x\right)  =I(G\circ H-v;x)\text{ and }q\left(  x\right)  =x\bullet
I(G\circ H-N_{G\circ H}\left[  v\right]  ;x).
\]

\textit{Claim1}. $p\left(  x\right)  $ is symmetric, unimodal with a unique
mode, equal to $n$.

First, $p\left(  x\right)  =I(H;x)\bullet I\left(  \left(  G-v\right)  \circ
H;x\right)  $, because $G\circ H-v$ is the disjoint union of $H$ and $\left(
G-v\right)  \circ H$.

The graph $G\circ H$ has a unique maximum independent set, namely the set $S$
containing the non-adjacent vertices from each $H$. Hence, we get that
$\alpha\left(  G\circ H\right)  =\alpha\left(  G\circ H-v\right)  =2n$, since
$S\cap V\left(  G\right)  =\emptyset$, and $v\in V\left(  G\right)  $.
According to Theorem \ref{th3}\emph{(ii)}, the polynomial $I\left(  \left(
G-v\right)  \circ H;x\right)  $ is symmetric, and by induction hypothesis,
$I\left(  \left(  G-v\right)  \circ H;x\right)  $ is also unimodal and its
mode is unique and equal to $n-1$. According to Lemma \ref{lem2}, $p\left(
x\right)  $ is symmetric and unimodal with a unique mode, equal to $n$.

\textit{Claim 2}. $q\left(  x\right)  $ is symmetric and unimodal.

Since $G\circ H-N_{G\circ H}\left[  v\right]  $ consists of the disjoint union
of $kH$ and $\left(  G-N_{G}\left[  v\right]  \right)  \circ H$, we obtain
that $q\left(  x\right)  =x\bullet\left(  I(H;x)\right)  ^{k}\bullet I\left(
\left(  G-N_{G}\left[  v\right]  \right)  \circ H;x\right)  $.

Further, one can see that $\alpha\left(  G\circ H-N_{G\circ H}\left[
v\right]  \right)  =\left\vert S-\left\{  a_{1},a_{2}\right\}  \right\vert
=\alpha\left(  G\circ H\right)  -2$, where $\left\{  a_{1},a_{2}\right\}
=S\cap N_{G\circ H}\left(  v\right)  $. The symmetry of $I\left(  \left(
G-N_{G}\left[  v\right]  \right)  \circ H;x\right)  $ follows from Theorem
\ref{th3}\emph{(ii)}. By induction hypothesis, $I\left(  \left(
G-N_{G}\left[  v\right]  \right)  \circ H;x\right)  $ is unimodal with a
unique mode. Lemma \ref{lem2} ensures that $\left(  I(H;x)\right)  ^{k}\bullet
I\left(  \left(  G-N_{G}\left[  v\right]  \right)  \circ H;x\right)  $ is
symmetric and unimodal, with a unique mode.

\textit{Claim 3}. $I(G\circ H;x)$ is symmetric, unimodal and its mode is
unique and equals $n$.

Since $\deg p=\deg q+1\geq2$, and $p\left(  0\right)  =1$, while $q\left(
0\right)  =0$, we finally obtain that $I(G\circ H;x)=p\left(  x\right)
+q\left(  x\right)  $ is symmetric and unimodal with a unique mode, according
to Lemma \ref{lem1}.
\end{proof}

Since $K_{2}-e=2K_{1}$, we obtain the following.

\begin{corollary}
\cite{Mandrescu2012} The polynomial $I\left(  G\circ2K_{1};x\right)  $ is
unimodal for every graph G.
\end{corollary}

\section{Conclusions\qquad}

In this paper we started investigating higher symmetries of polynomials with
emphasis on independence polynomials of graphs. This new paradigm already
showed its usefulness in revealing a new family of graphs with symmetric
independence polynomials. We conclude with the following.

\begin{conjecture}
$I\left(  G\circ H;x\right)  $ is symmetric for every graph $G$ if and only if
$H=K_{r}-e$ for some $r\geq2$.
\end{conjecture}

\begin{problem}
Describe the set of invariants of a given rational function $\frac{a\left(
x\right)  }{b\left(  x\right)  }$.
\end{problem}


\begin{thebibliography}{99}                                                                                               %


\bibitem {AlMaScEr87}Y. Alavi, P. J. Malde, A. J. Schwenk, P. Erd\"{o}s,
\emph{The vertex independence sequence of a graph is not constrained},
Congressus Numerantium \textbf{58} (1987) 15-23.

\bibitem {Andrews}G. E. Andrews, \emph{The Theory of Partitions},
Addison-Wesley, Reading, 1976.

\bibitem {BS2010}P. Bahls, N. Salazar, \emph{Symmetry and unimodality of
independence polynomials of path-like graphs}, The Australasian Journal of
Combinatorics \textbf{47} (2010) 165-176.

\bibitem {Berge0}C. Berge, \emph{F\"{a}rbung von Graphen deren s\"{a}mtliche
bzw. deren ungerade Kreise starr sind (Zusammenfassung)}, Wiss.Z.
Martin-Luther-Univ. Halle \textbf{10} (1961) 114-115.

\bibitem {ChRoSeTh03}M. Chudnovsky, N. Robertson, P. D. Seymour and R. Thomas,
\emph{Progress on perfect graphs}, Mathematical Programming B \textbf{97}
(2003) 405-422.

\bibitem {ChRoSeTh04}M. Chudnovsky, N. Robertson, P. D. Seymour and R. Thomas,
\emph{The Strong Perfect Graph Theorem}, Annals of Mathematics \textbf{164}
(2006), 51-229.

\bibitem {GuHa83}I. Gutman, F. Harary, \emph{Generalizations of the matching
polynomial}, Utilitas Mathematica \textbf{24} (1983) 97-106.

\bibitem {Gu92c}I. Gutman, \emph{Independence vertex palindromic graphs},
Graph Theory Notes of New York Academy of Sciences \textbf{XXIII} (1992) 21-24.

\bibitem {Gu92d}I. Gutman, \emph{Independence vertex sets in some compound
graphs}, Publications de l'Institut Math\'{e}matique \textbf{52} (1992) 5-9.

\bibitem {Gu93}I. Gutman, \emph{A contribution to the study of palindromic
graphs}, Graph Theory Notes of New York Academy of Sciences \textbf{XXIV}
(1993) 51-56.

\bibitem {Kennedy}J. W. Kennedy, \emph{Palindromic graphs}, Graph Theory Notes
of New York Academy of Sciences \textbf{XXII} (1992) 27-32.

\bibitem {LeMa03b}V. E. Levit, E. Mandrescu, \emph{On unimodality of
independence polynomials of some well-covered trees}, DMTCS 2003 (C. S. Calude
\textit{et al}. eds.), LNCS \textbf{2731}, Springer-Verlag (2003) 237-256.

\bibitem {LeMa03c}V. E. Levit, E. Mandrescu, \emph{A family of well-covered
graphs with unimodal independence polynomials}, Congressus Numerantium
\textbf{165} (2003) 195-207.

\bibitem {LeMa04c}V. E. Levit, E. Mandrescu, \emph{Very well-covered graphs
with log-concave independence polynomials}, Carpathian Journal of Mathematics
\textbf{20} (2004) 73-80.

\bibitem {LevManGreece}V. E. Levit, E. Mandrescu, \emph{The independence
polynomial of a graph - a survey}, Proceedings of the 1$^{st}$ International
Conference on Algebraic Informatics, Aristotle University of Thessaloniki,
Greece, (2005) 233-254. http://web.auth.gr/cai05/papers/20.pdf

\bibitem {LeMa04b}V. E. Levit, E. Mandrescu, \emph{Independence polynomials of
well-covered graphs: Generic counterexamples for the unimodality conjecture},
European Journal of Combinatorics \textbf{27} (2006) 931-939.

\bibitem {LevMan2007b}V.E. Levit, E. Mandrescu, \emph{Independence polynomials
and the unimodality conjecture for very well-covered, quasi-regularizable, and
perfect graphs}, Graph Theory in Paris, Trends Math., Birkh\"{a}user, Basel
(2007) 243-254.

\bibitem {LevMan2007}V. E. Levit, E. Mandrescu, \emph{A family of graphs whose
independence polynomials are both palindromic and unimodal}, Carpathian
Journal of Mathematics \textbf{23} (2007) 108-116.

\bibitem {LevMan2008}V. E. Levit, E. Mandrescu, \emph{Graph operations and
partial unimodality of independence polynomials}, Congressus Numerantium
\textbf{190} (2008) 21-31.

\bibitem {LevMan2010}V. E. Levit, E. Mandrescu, \emph{On the independence
polynomial of an antiregular graph}, Carpathian Journal of Mathematics
\textbf{28} (2012) 263-272.

\bibitem {Mandrescu2012}E. Mandrescu, \emph{Unimodality of some independence
polynomials via their palindromicity}, Australasian Journal of Combinatorics
\textbf{53} (2012) 76-82.

\bibitem {St98}D. Stevanovi\'{c}, \emph{Graphs with palindromic independence
polynomial}, Graph Theory Notes of New York Academy of Sciences \textbf{XXXIV}
(1998) 31-36.

\bibitem {Wang}Yi Wang, Bao-Xuan Zhu, \emph{On the unimodality of independence
polynomials of some graphs}, European Journal of Combinatorics \textbf{32}
(2011) 10-20.

\bibitem {Zfu}Z.F. Zhu, \emph{The unimodality of independence polynomials of
some graphs}, Australasian Journal of Combinatorics \textbf{38} (2007) 27-33.
\end{thebibliography}
\end{document}